\documentclass[12pt]{article}

\usepackage{authblk}
\usepackage{amssymb}
\usepackage{amsmath}
\usepackage{mathrsfs}
\usepackage{amsfonts}
\usepackage{braket}
\usepackage{color}
\usepackage{graphicx}
\usepackage{tikz}
\usetikzlibrary{shapes,snakes}
\usepackage{pgfplots}
\pgfplotsset{compat=1.9} 
\newtheorem{thm}{Theorem}

\newtheorem{lem}[thm]{Lemma}
\newtheorem{prop}[thm]{Proposition}

\newtheorem{defn}[thm]{Definition}
\newcommand{\norm}[1]{\left\Vert#1\right\Vert}
\newcommand{\abs}[1]{\left\vert#1\right\vert}

\def\XXint#1#2#3{{\setbox0=\hbox{$#1{#2#3}{\int}$ }
\vcenter{\hbox{$#2#3$ }}\kern-.6125\wd0}}

\newcounter{lastnote}

\title{Covariant Fermionic Fields of Space-like Particles}
\author{Radhakrishnan Balu}
\affil{Army Research Laboratory Adelphi, MD, 21005-5069, USA \\
       radhakrishnan.balu.civ@mail.mil
      }  
\affil{Department of Mathematics \&
	Norbert Wiener Center for Harmonic Analysis and Applications, 
	University of Maryland,
	College Park, MD 20742.\\
	rbalu@umd.edu}          
\date{Received: date / Accepted: \today}
\begin{document}
\maketitle

\begin{abstract}
 We develop covariant fermionic fields of space-like particles. As an application of the formalism we discuss the example of superluminous tachyons with imaginary rest mass and spin 1/2 forming fermionic ensembles that is relativistically covariant. We use the tachyonic single particle Hilbert space to build the quantum stochastic fields as this notion is more consensual, instability in vacuum quantum field than faster than light particles, in field theory with applications to superconductivity via tachyonic condensation. We also establish the localizability and as a consequence the causality of tachyons (in terms of anticommutators) using the systems of imprimitivity framework and derive the Dirac equation for the space-like particle that has an additional helicity term.  
\end{abstract}

\section{Introduction}
\label{intro}
 In our earlier work we treated covariant Quantum Fields via Lorentz Group Representation of Weyl Operators \cite {Rad2020}, \cite {Rad2020b} constructed Fock spaces for massless Weyl fermions and covariant white noises. Here, we develop methods for second quantizing imaginary mass fermions conforming to special theory of relativity that are space-like particles, that is based on systems of imprimitivity, to field theoretic context. Systems of imprimitivity \cite {WignerLz1939}, \cite{Mackey1963}, \cite {Varadarajan1985} is a powerful characterization of dynamical systems, when the configuration space of a quantum system is described by a group, from which infinitesimal forms in terms of differential equations ($Shr\ddot{o}dinger$, Heisenberg, and Dirac etc), localizability, and the canonical commutation relations can be derived \cite {Wigner1949} and \cite {Rad2019}.  In Varadarajan{'}s work \cite {Varadarajan1985} space-like particles are not considered as imaginary mass particles with speed greater than light and the non existence of finite unitary representation for the little group created a lack of interest in general. However, there are some recent work in studying tachyons \cite {Schwartz2016} and overcoming the group representation problem that inspired our work. Let us look at the 2 + 1 spacetime case considered in this work as the little group in this case is O(2, 1) with the Lie algebra generators satisfying the following conditions:
\begin {equation}
 [J_z, K_x] = -ik_y; [J_z, K_y] = -ik_x; [K_y, K_x] = -iJ_z;
\end {equation}
Where the generators are given by the Pauli matrices as:
\begin {equation}
J_z = 1/2 \sigma_3; K_x = i/2 \sigma_1; K_y = i/2 \sigma_2;
\end {equation}

An operator of the form $U = e^{it(\alpha J_z + \beta K_x + \gamma K_y)}$ is in general not unitary. However, as Schwartz \cite {Schwartz2016} observed the relation 
\begin {align}
U^\dag &= U \sigma_3 U^{-1} \text { in conjunction with the metric } \\ 
\langle \psi, \sigma_3\phi \rangle \rightarrow \langle U\psi, \sigma_3 U\phi \rangle &= \langle \psi, U^\dag \sigma_3 U \phi \rangle = \langle \psi, \sigma_3\phi \rangle.
\end {align}
Another way to look at the above expression is to use $\sqrt{\sigma_3}$ which is also unitary and can be absorbed in to U as
\begin {align}
\langle U\psi, \sigma_3 U\phi \rangle &= \langle \sqrt{\sigma_3}^\dag U\psi, \sqrt{\sigma_3} U\phi \rangle. \\
\sqrt{\sigma_3} &= \begin{bmatrix} 1 & 0 \\ 0 & i \end{bmatrix}.
\end {align}
Borrowing ideas from quantum information processing it is like rotating all the components of the fibers, each fiber would be represented by two qubits, by the same angle that leaves the systems invariant.

As we will see later a family of metrics including the above one arise naturally in describing the covariant spinor fields in our framework. Formally, a spinor field is a map from a G-space where Poincar$\grave{e}$ group acts to the spin group representation. Schwartz considers the operator $\sigma_3$ as helicity and as a way to distinguish particles against anti-particles as the mass shell hyperbola in the case of space-like particles is a single sheet one. We will take this consideration in building our fiber bundle to describe the Hibert space of the space-like particle.Our frame is based on orthomoduar lattices, quantum logic, where the above inner product can arise as the duality $\eta$ which is a $\theta-$bilinear form satisfying:
\begin {align}
\langle x_1 + x_2, y \rangle &= \langle x_1, y \rangle + \langle x_2, y \rangle. \\
\langle x_1 , y_1 + y_2 \rangle &= \langle x, y_1 \rangle + \langle x, y_2 \rangle. \\
\langle cx, dy \rangle &= c \langle c, y \rangle d^\theta.
\end {align}
To elaborate further on this, by choosing different values of $\theta$ different inner product can be defined. A polarity $\eta$ is duality which is also involutive, $(\eta(\eta(x)) = x)$, and orthocomplemented lattices are endowed with a polarity. Finally, we note that covariant metrics are required to describe spinor fields, as they are not topological, of both superluminous (faster than light) and subluminous (slower than light) particles.

Our quantization is based on SI and an important theorem by Mackey that characterizes such systems in terms of induced representations, and so let us recollect key notions in Clifford algebras, spinor fields, and Schwartz spaces \cite {Varadarajan1985} before discussing our main results on covariant tachyonic fields and localizability of the particles by allowing imaginary values for momentum \cite {Jent2012}.

\section {Little groups (stabilizer subgroups)}
Some of the different systems of imprimitivity that live on the orbits of the stabilizer subgroups are described below. It is good to keep in mind the picture that SI is an irreducible unitary representation of Poincar$\grave{e}$ group $\mathscr{P}^+$ induced from the representation of a subgroup such as $SO_3$, which is a subgroup of homogeneous Lorentz, as $(U_m(g)\psi)(k) = e^{i\{k,g\}}\psi(R^{-1}_mk)$ where g belongs to the $\mathscr{R}^4$ portion of the Poincar$\grave{e}$ group, m is a member of the rotation group, and the duality between the configuration space $\mathbb{R}^4$ and the momentum space $\mathbb{P}^4$ is expressed using the character the irreducible representation of the group $\mathbb{R}^4$ as:
\begin {align} \label {eq: poissonBr}
\{k,g\} &= k_0 g_0 - k_1 g_1 - k_2 g_2 - k_3 g_3, p \in \mathbb{P}^4. \\
\hat{p}:x &\rightarrow e^{i\{k,g\}}. \\
\{Lx, Lp \} &= \{ x, p \}. \\\
\hat{p}(L^{-1}x) &= \hat{Lp}(x).
\end {align}
In the above L is a matrix representation of Lorentz group acting on $\mathbb{R}^4$ as well as $\mathbb{P}^4$ and it is easy to see that $p \rightarrow Lp$ is the adjoint of L action on $\mathbb{P}^4$. The $\mathbb{R}^4$ space is called the configuration space and the dual $\mathbb{P}^4$ is the momentum space of a relativistic quantum particle.

The stabilizer subgroup of the Poincar$\grave{e}$ group $\mathscr{P}^+$ (Space-like particle) can be described as follows: \cite {Kim1991}:
The Lorentz frame in which the walker is at rest has momentum proportional to (0,1,0,0) and the little group is
again SO(3, 1) and this time the rotations will change the helicity. 
\
\section {Smooth Orbits in Momentum Space}
We need to described few ingredients to construct the Hilbert space of space-like fermions namely, the fiber bundle, the fiber vector space, an inner product for the fibers, and an invariant measure. 
The 3+1 spacetime Lorentz group $\hat{O}(3,1)$-orbits  of the momentum space $\mathscr{R}^4$, where the systems of imprimitivity established will live, described by the symmetry $\hat{O}(3,1)\rtimes\mathbb{R}^4$. The orbits have an invariant measure $\alpha^+_m$ whose existence is guaranteed as the groups and the stabilizer groups concerned are unimodular and in fact it is the Lorentz invariant measure $\frac{dp}{p_0}$ for the case of forward mass hyperboloid.  The orbits are defined as:
\begin {align}
Y^{+,1/2}_{im} &= \{p: p^2_0 - p^2_1 - p^2_2 - p^2_3 = -m^2, p_0 > 0\}, \text{\color{blue} single sheet hyperboloid}.
\end {align}
Each of these orbits are invariant with respect to $\hat{O}(3,1)$ and let us consider the stabilizer subgroup of the first orbit at p=(0,1,0,0). Now, assuming that the spin of the particle is 1/2 and mass $m > 0$  let us define the corresponding fiber bundles (vector) for the positive mass hyperboloid that corresponds to the positive-energy states by building the total space as a product of the orbits and the group $SL(2, C)$.
\begin {align} \label {eq: bundle}
\hat{B}^{+,1/2}_{im} &= \{(p,v) \text{   }p\in{\hat{Y}^{+1/2}_{im}, }\text{   }v\in\mathscr{C}^4,\sum_{r = 0}^3 p_r \gamma_r v = mv\}. \\
\hat{\pi} &: (p,v) \rightarrow {p}. \text{  Projection from the total space }\hat{B}^{+,2}_{im} \text{ to the base }\hat{Y}^{+,1/2}_{im}.
\end {align}
It is easy to see that if $(p, v) \in B_{im}^{+,1/2}$ then so is also $(\delta(h)p, S(h^{*-1})v)$. Thus, we have the following Poincar$\grave{e}$ group symmetric action on the bundle that encodes spinors into the fibers:
\begin {equation}
h,(p,v) \rightarrow (p,v)^h = (\delta(h)p, S(h^{*-1})v).
\end {equation}

\begin {prop} There is a fiber representation of spin 1/2 space-like particles with spinor fields supporting a family of covariant inner products.

For $m > 0$ the fiber of $B_{im}^{+,1/2}$ at $p^(m) = (0, m, 0, 0)$ is spanned by the vectors $\{e_1 + ie_4, e_2 + ie_3\}$ and so $B_{im}^{+,1/2}$ is two dimensional with the covariant form $-im^{-1}\langle\gamma_3 v,v\rangle$. This positive definite due to skew symmetric $\gamma_3$ and $-im^{-1}\langle\gamma_3 v,v\rangle$ is purely imaginary. This 1-form is covariant with respect to the spin representation $S(h)^*\gamma_3 S(h) = \gamma_3$. As we mentioned earlier, both superluminal and subluminal particles support such a family of covariant 1-forms and thus metrics in describing their spinor fields. The same set of arguments can be applied to $B_{-im}^{+,1/2}$ implying that the bundles $B_{-im}^{+,1/2}$ also converge to $B_0^+$.
The endomorphism (chirality or helicity operator) $\Gamma = i\gamma_0\gamma_1\gamma_2\gamma_3$ transforms $B_{im}^{+,1/2}(p) \rightarrow B_{-im}^{+,1/2}(p), \forall p \in X_{im}^+, m > 0$ as it anticommutes with all the gammas. In the limit $\Gamma$ leaves the fibers of $B_0^+$ invariant leading to higher degeneracies with $\Gamma = \begin{bmatrix} 1 & 0 \\ 0 & -1 \end{bmatrix}$. This means $\Gamma$ commutes with all of $S(h)$ implying that $(p, v) \in B_0^+ \Rightarrow (p, \Gamma v) \in B_0^+$. 

$\Gamma$ has eigen values $\pm 1$ at fiber (0, 1, 0, 0) and hence is true of all fibers. The stability group $E^*$ at (0, 1, 0, 0) given by the generators mentioned during the introduction.
is still invariant and we have $E_2$ is the stabiliser group at (0,1,0,0).

Now, we can define the states of the space-like particles on the Hilbert space $\hat{\mathscr{H}}^{+,1/2}_{im}$, square integrable functions on Borel sections of the bundle $\hat{B}^{+,1/2}_{im} = \{(p,v) : (p,v) \in B_{im}^{+,1/2}, \Gamma v = \mp v\}$.

with respect to the invariant measure $\beta^{+,1}_{im}$, whose norm induced by the inner product is given below:
\begin {equation} \label {eq: section}
 \norm{\phi}^2 = \int_{X^+_m}p_3^{-1}\langle\phi{p},\phi{p}\rangle.{d\beta}^{+,1/2}_{im}(p).
\end {equation}
The invariant measure and the induced representation of the Poincar$\grave{e}$ group from that of the Weyl fermion are given below:
\begin {align}
{d\beta}^{+,1/1}_{im}(p) &= \frac {dp_1 dp_2 dp_3} {2(-m^2 + p_1^2 + p_2^2 + p_3^2)}, m > 0, \text{imaginary values for p are allowed}. \\
(U_{h,x}\phi)(p) &= \text{exp i}\{x, p\} \phi (\delta(h)^{-1}p)^h.
\end {align}
\end {prop}
\
We need Schwartz spaces, and their duals the tempered distributions
to guarantee Fourier transforms, to move with ease between configuration and momentum
space descriptions for Dirac equation as well as localizability and so let us define them.
\begin {defn} Let V be an n-dimensional Euclidean space with the positive definite product $\langle .,.\rangle$ and Lebesgue measure dv. for any translation invariant differential operator D and any complex polynomial q on V, the function $\phi \rightarrow sup_{v \in V} \abs{q(v)D(\phi)(v)}$ is a seminorm on $C_c^\infty(V)$ and the collection of these seminorms induces a locally convex topology for $C_c^\infty(V)$. Its completion may be identified with $\mathscr{S}(V)$, the space of $C^\infty$ functions on V for which $sup_{v\in V} \abs{q(v)D(\phi)(v)} < \infty$ for all q and D. Intuitively, it is the function space of all infinitely differentiable functions that are rapidly decreasing at infinity along with all partial derivatives 

Let $\{v_1, \dots, v_n\}$ be an orthonormal basis for V and $D_{v_1, \dots, v_n} = (\partial/\partial x_1)^{\nu_1} \dots (\partial/\partial x_n)^{\nu_n}$ linearly span the algebra of translation invariant differential operators, and $x_1, \dots, x_n$ be the linear coordinate functions on V associated with the chosen basis. A topology on $\mathscr{S}(V)$ can be induced by the collection of seminorms \\
$sup_{x_1 \dots , x_n} abs{(1+x_1^2 + \dots + x_n^2)^k(D_{v_1, \dots, v_n}\phi)(x_1, \dots , x_n)}$ for various k=0,1,2,..., and \\
$\nu_1, \dots, \nu_n \ge 0$.  A tempered distribution E is a complex valued linear functional on $C_c^\infty(V)$ which is continuous in the topology defined. By extending these functionals to $\mathscr{S}(V)$ we may regard a tempered distribution as a continuous linear functional on the Schwartz space $\mathscr{S}(V)$.
\end {defn}
Fourier transform is an an automorphism on Schwartz space as
\begin {equation}
\hat{\phi(x)} = (2\pi)^{-n/2}\int_V exp[-i \langle x,v\rangle]\phi(v)dv, \forall \phi \in \mathscr{S}(V), x \in V.
\end {equation}
The Fourier transform of a tempered distribution $E(\phi \rightarrow E(\phi), \phi \in \mathscr{S}(V))$ is given by $\hat{E}(\phi) = E(\hat{\phi})$. We denote the dual space of $\mathscr{S}(V)$ Schwartz functions that consists of the tempered distributions as $\mathscr{S}'(V)$ and they form the Gelfand nuclear triple $\mathscr{S}(\mathscr{H}_0^{+2}) \subset \mathscr{H}_0^{+2} \subset \mathscr{S}'(\mathscr{H}_0^{+2})$.

Let us now state and establish the main result for the case of massive Weyl fermion with space-like momentum from the representation of a subgroup (space-like little group) of Poincar\`e.

\begin {thm} Space-like Weyl representation of Poincar$\grave{e}$ group is a transitive system of imprimitivity. 
\end {thm}
\begin {proof}
Let $g:\mathscr{G} \rightarrow U_g(\hat{\mathscr{H}}^{+,1/2}_{im})$ be a homomorphism from the two dimensional Euclidean group $\mathscr{G} = E_2$ to the unitary representation of the group in $\hat{\mathscr{H}}^{+,1/2}_{im}$. We note that it is a stabilizer subgroup which is also closed at the momentum (0,0,m,0) and so $H/\mathscr{G}$ is a transitive space.

Consider a map $v(g):\mathscr{G} \rightarrow \hat{\mathscr{H}}^{+,1/2}_{im}$ satisfying the first order cocycle relation $v(gh) = v(g) + U_g v(h), g,h \in \mathscr{G}$.
An example of such a map is the following: \cite {KP1992} 
\begin {align*}
\mathscr{H} &= \oplus_{j=0}^\infty \mathscr{H}_j. \\
H &= 1 \oplus \oplus_{j=1}^\infty H_j. \\
U_t &= e^{-itH}, t \in \mathscr{G}. \\
v(t) &= tu_0 \oplus \oplus_{j=1}^\infty (e^{-itH_j} u_j
- u_j).
\end {align*}

Now, we can define the Weyl operator $V_g = W_g (v(g), U_g)$ where $g \in \mathscr{G}$ in the Fock space $\Gamma_s(\hat{\mathscr{H}}^{+,1/2}_{im})$. 

 This is a projective unitary representation satisfying the commutator relation $V_g V_h = e^{iIm\langle v_g, U_g v_h \rangle} V_h V_g$ and let us denote the homomorphism from $\mathscr{G}$ to $V_g$ by m. This guarantees a map b that satisfies $b(gh) = b(g)m(h), g \in \mathscr{P}, h \in \mathscr{G}$ and such map can be constructed by considering the map
$c(x \rightarrow c(x)$ as Borel section of $\mathscr{P} / \mathscr{G}$ (the choice of this section not a canonical one but immaterial to our purpose here) with $c(x_0) = e$. The map
$\beta$ maps $g \in \mathscr{P} \rightarrow g \mathscr{G}$
\begin {align}
a(g) &= c(\beta (g))^{-1}. \\
b(g) &= m(a(g)). 
\end {align}
Then the strict cocycle $\phi$ satisfies $\phi(g_1, g_2) = b(g_1 g_2)b(g_2)^{-1}$.

We can now set the SI relation using the above cocycle as follows:
\begin {align*}
U_{h, x}\phi(p) &= exp i\{x, p\} \phi(g, g^{-1}x) f(g^{-1}x),  f\in\mathscr{H} \\
&\text {  character representation is defined in equation } \eqref{eq: poissonBr}.\\
P_E F &= \chi_E f \text { Position  operator}. \\
\end {align*}

We can construct the conjugate pair of field operators for the Fock space $\Gamma_s(\hat{\mathscr{H}}^{+,1/2}_0))$ as follows:

Let $p_g$ be the stone generator for the family of operators $P_{gt,p}, g \in \mathscr{P}, t \in \mathbb{R}$ and q(g) = p(ig) and we get the creation and annihilation operators as $a(g)^\dag = \frac {1}{2} ( q(g) - i p(g)$ and  $a(g) = \frac {1}{2} ( q(g) + i p(g)$.

Strictly speaking we need to build fermionic Fock space and set up the corresponding annihilation and creation operators. Instead, we have constructed the bosonic Fock space, as there is no counterpart for Weyl second quantized operators in fermionic Fock space,  and using the unification of the \cite {Hudson1986} two spaces, $dB= (-1)^\Lambda dA$, we can construct the fermionic processes. This involves fashioning reflection processes and we leave out the details as our main focus here is the white noise process.
$\blacksquare$
\end {proof}

\section {Localizability}
We establish the localizability of tachyons of spin 1/2 by constructing a projection valued measure (PVM) as part of the SI relation. As our main theorem establishes an SI that comes with a PVM on spacetime $\mathscr{R}^4$. We would like to have a PVM on $\mathscr{R}^3$ to signify localization on space only degrees of freedom. We modify the proof from Varadarajan{'}s book for the particle with real mass to tachyons of theorem 9.16. This theorem is basis for establishing the fact that photons are not localizable. The main line of arguments run like this: we consider a Lorentz transformatiom from an inertial frame that leads to a configuration separated space-only coordinates. We can denote that unitary transformation by $U_{h,a}, a = a_0, \hat{a}$ and call this subgroup of Poincar$\grave{e}$ as M. As series of lemmas on equivalent induced representations on M and necessary and sufficient conditions for the existence of PVM on $\mathscr{R}^3$ culminates in the following theorem.
\begin {thm} \cite {Varadarajan1985} Let the representation U associated with a system $\mathscr{S}$ be defined by a pair $(X^+_m, \pi)$, where $m \ge 0$, and $\pi$ is a representation of the stability subgroup of $L^*$ at the point p = (m, 0, 0 ,0) when m > 0 and (1, 0, 0, 1). The system is localizable when m > 0. If m = 0, $\mathscr{S}$ is localizable if and only if there exists a representation $\pi'$ of $K^*$ such that $\pi$ and $\pi'$ have equivalent restrictions to the group $T'$ of all diagonal matrices of $K^*$.
\end {thm}
As we mentioned earlier we want to keep in mind there is a finite unitary representation for the space-like little group using the generalized inner product with Hermitian matrix $\sigma_3$. This theorem can be modified for the case of tachyons with non zero mass by considering the stability subgroup at p = (0, 0, 0, m) and selecting a strict $(H^*, X_{im})$-cocycle
C with values in the Hilbert space $\mathscr{K}$ in which $\pi$ acts such that C defines the representation $\pi$ at p. The unitary representation U acts in $\mathscr{H} = \mathscr{L}^2(X_{im}, \mathscr{K}, \beta^+_{im})$, and for $(h, a) \in M$ we have
\begin {equation}
(U_{h, a}
f)(p) = exp[-i\langle a, p \rangle]C(h, \delta(h)^{-1}p)f(\delta(h)^{-1}p).
\end {equation}
The proof for the above theorem depends upon a lemma 9.12 \cite {Varadarajan1985} that provides the necessary and sufficient conditions for existence of projection valued measures.
\begin {lem} Let V be a representation of M in some Hilbert space $\mathscr{H}$. Let m be a real number greater or equal to zero. Then, the existence of a projection valued measure P based on $R^3$ such that (V, P) is a system of imprimitivity, it is necessary and sufficient that there exists a representation s, $s \rightarrow s(h)$ of $K^*$ in Hilbert space $\mathscr{K}$ such that V is equivalent to the representation $V^*$ of M in $\mathscr{L}^2(R^3, \mathscr{K}, \alpha_m)$ defined by 
\begin {align}
(V^s_{h, a}f)(y) &= exp[-i\langle a, y \rangle]s(h)f(\delta(h)^{-1}y). \\
d\alpha_m &= \frac {dy} {2(m^2 + \langle y, y \rangle)^{1/2}}.
\end {align}
\end {lem}

\section {Dirac equation}
Let us rewrite the equation \ref {eq: bundle} in the single particle Hilbert space to match the final form available in the literature \cite {Jent2012} and then apply the Fourier transform as we are working with Schwartz spaces the transform is well defined.
\begin {align}
\sum_{r = 0}^3 p_r \gamma_r v &= mv\}, \Gamma v = \mp v. \\
(\sum_{r = 0}^3 ip_r \gamma_r - \Gamma m)v &= 0. \Rightarrow \\
(\sum_{r = 0}^3 i \frac {\partial} {\partial x_r}  \gamma_r - \Gamma m)\Phi &= 0. \text { Schwartz spaces allow Fourier transform}
\end {align}

\section {Causality for Dirac Tachyons} Let us now look at how causality arise from SI relations on space-like particles lifted to second quantization for a simple situation. As we have constructed Weyl operators that form a system of imprimitivity (V, P) on the space $L^2((E^*), \mu) \cong \Gamma_s(\hat{\mathscr{H}}^{+,1/2}_{im})$ by induced an representation of the  Poincar$\grave{e} \mathscr{P}^* = H^*\rtimes \mathbb{R}^4 $ group we have
\begin {align*}
\alpha \rightarrow U_\alpha  &  \\
(U_\alpha f)(x) &= x(a) f(x),  (x \in \mathbb{R}^4). \\
(W_{ah}f)(x) &= x(a)(V_h f)(x), (x \in \mathbb{R}^4). \\
U_\alpha \phi(x) U^{-1}_\alpha &= \phi(\alpha x),  \phi(x) \in \Gamma_s(\hat{\mathscr{H}}^{+,1/2}_{im}).
\end {align*}
Let us start with the anti-commutator of Schwartz distributions as
$[\phi(x), \psi(x)]_+ = \delta(x, y), \phi, \psi \in \Gamma_s(\hat{\mathscr{H}}^{+,1/2}_{im})$ and apply the SI relation to get the anticommutation relation \cite {Schwartz2016}
\begin {align*}
[\phi(x), \psi(x)]_+ &= \delta^3(x, y). \\
\Rightarrow [U_\alpha\phi(x)U^{-1}_\alpha, U_\alpha\psi(x)U^{-1}_\alpha]_+ &= \delta^3(\alpha x, \alpha y) = [\phi(x), \psi(x)]_+ = \delta^3(x, y). \\
\Rightarrow [\phi(x), \psi(x)]_+ &= \delta^3(x - y).
\end {align*}

\
\section {Summary and conclusions}
We derived the covariant field operators for space-like particles using induced representations of groups and  expressed them in terms of systems of imprimitivity. We established the results for the massless Weyl fermion case by inducing a representation of Poincar$\grave{e}$ group from the subgroup that is a stabilizer at the momentum (0,0,1,0). Finally, we derived  the Dirac equation and established localizability for tachyons.


\end{document}